%% file: paper.tex
\documentclass[11pt,a4paper]{article}
\usepackage[utf8]{inputenc}
\usepackage{amsmath}
\usepackage{amsfonts}
\usepackage{amssymb}
\usepackage{graphicx}
\usepackage[ruled]{algorithm}
\usepackage{algorithmicx}
\usepackage{algpseudocode}
\usepackage{stackrel}	% for triple parameters in stackrel
\usepackage{enumitem}	% for enumerate
\usepackage[title]{appendix}
\usepackage{placeins}	% for floatbarrier
\usepackage{rotating}	% includes \sidewaysfigure
\usepackage{verbatim} 
\usepackage{amsthm}		% for proof
\usepackage{tcolorbox}	% for the boxes with implementation explanations
\usepackage{hyperref}
\usepackage{url}
\usepackage{subcaption}
\usepackage{soul}
\usepackage{tabularx}
\usepackage{multirow}
\usepackage{subcaption}
\usepackage{authblk}

\author{Ricardo Almeida}
\affil{Keysight Laboratories, Keysight Technologies}
\title{A Report on Achieving Complete Regular-Expression Matching \\ using Mealy Machines}
\date{May 27, 2021}

\DeclareFontFamily{U}{mathx}{\hyphenchar\font45}
\DeclareFontShape{U}{mathx}{m}{n}{ <-> mathx10 }{}
\DeclareSymbolFont{mathx}{U}{mathx}{m}{n}
\DeclareFontSubstitution{U}{mathx}{m}{n}
\DeclareMathAccent{\widebar}{\mathalpha}{mathx}{"73}

\input{defns.tex}

\begin{document}

\maketitle

\begin{abstract}
While regexp matching is a powerful mechanism for finding patterns in data streams,
regexp engines in general only find matches that do not overlap.
Moreover, different forms of nondeterministic exploration, 
where symbols read are processed more than once, are often used,
which can be costly in real-time matching.
We present an algorithm that constructs from any regexp a Mealy machine that finds \textbf{all} matches and while reading each input symbol only once.
The machine computed can also detect and distinguish different patterns or sub-patterns inside patterns.
Additionally, we show how to compute a minimal Mealy machine via a variation of DFA minimization,
by formalizing Mealy machines in terms of regular languages.
\end{abstract}

\section{Introduction}		\label{sec:introduction}

Regular-expression (regexp) matching is a powerful mechanism with a wide variety of applications,
such as lexical analysis by a compiler~\cite{Johnson1968},
pattern matching in text files~\cite{Thompson1968},
network intrusion detection~\cite{ceska2019,Bispo2007},
pattern detection in time-series data~\cite{Rodrigues2019},
 etc.
There are many powerful regexp engines available and today most programming languages also support some form of regexp matching.
However, engines in general do not find all matches (also known as instances) of the regexp in a given string, 
even when reporting to do so.
More precisely, 
engines commonly only find instances that do not overlap,
which can happen (1) when the last symbols of an instance in the sequence are at the same time the first symbols of the next instance found of the same pattern, or 
(2) when the regexp encodes more than one pattern to be detected and some instances of one pattern overlap with (or are even entirely contained in) instances of another pattern.

There are several real-world applications where supporting (1) and (2) is important, 
with (1) being necessary in some cases.
For a simple example of (1), 
consider the digital-triggering problem of detecting pulse waves in an input signal.
A pulse is a rectangular waveform in which the amplitude alternates between fixed low and high values,
where the period from the low-to-high transition to the high-to-low transition can be counted as one pulse.
If we encode each time point in the signal's waveform as a symbol representing its amplitude ($l$ for low and $h$ for high),
then detecting pulses becomes a problem of looking for matches of the expression $lh^+l$ in the signal's encoding
(i.e., a low sample followed by 1 or more high samples and then a low sample).
When analysing waveforms like the square wave ($lhlhlhlh\ldots$),
regexp engines in general only find half of the instances,
and support for (1) is needed to find the other half.

\begin{table}[hbtp]
\footnotesize
\begin{center}
\begin{tabular}{l|l}
 \textbf{Pattern regexp}  & \textbf{Outputting condition} \\
\hline
 \multirow{2}{10em}{$e_1 = a(b+c)^+d\langle\alpha\rangle$} & Ouput $\alpha$ after any string with one $a$ (at the start), \\
& one $d$ (at the end) and one or more $b$'s or $c$'s in between. \\
\hline
 \multirow{3}{11.1em}{$e_2 = d((a^*b^++b^*)c)^+d \langle\beta\rangle$} & Ouput $\beta$ after any string with two $d$'s (one at the start  \\
 & and one at the end) where every sequence of $a$'s is followed \\ 
 &  by a $b$ and every sequence of $b$'s is followed by a $c$. \\
 \hline
 $e_3 = e_1 + e_2$  & Produce any outputs according to $e_1$ or $e_2$.
\end{tabular}
\end{center}
\caption{Examples of patterns that can be detected by FSMs.}
\label{table:three_expressions}
\end{table}

For a more elaborate example of (1), and also of (2),
consider the patterns defined by the regexps $e_1$ and $e_2$ from Table~\ref{table:three_expressions}.
We distinguish between occurrences of the two patterns by using different output symbols ($\alpha$ and $\beta$) that are to be emitted when matches of one or the other are detected.
These output symbols are specified in $e_1$ and $e_2$ after the corresponding input symbols and surrounded by '$\langle$' and '$\rangle$'.
Now consider the example matching trace 
\setuldepth{\large{$\langle\alpha\beta\rangle$}} 
$\! s = \overline{ab\mbox{\ul{$d$}}}\underline{\langle\alpha\rangle bc\underline{abcbc\overline{d\raisebox{2.89mm}{}}}}\overline{\langle\alpha,\beta\rangle cd} \langle\beta\rangle$,
%where output symbols produced are preceded by the input symbols that match the corresponding regexp.
where both $e_1$ and $e_2$ occur twice
(each instance is individually underlined/overlined for ease of readability).
%Normally, a regexp engine would detect both matches of $e_1$, since they do not overlap,
%but only the first of the matches of $e_2$, since they do overlap.
Consider now that matches must be identified as early as possible and efficiency overall is critical,
such as in real-time triggering, where the option to cache input for later processing is limited.
For this reason,
we need a model that detects all matches of $e_1$ and $e_2$ (which we call \emph{complete matching})
on any input and while processing each input symbol only once,
thus enabling us to process input data at the same rate it is read.

Traditionally, regexp engines perform the matching by first converting the regexp to an equivalent finite-state machine (FSM, such as a DFA that reads an input symbol at a time and either accepts or rejects the input read so far if it matches \emph{exactly} the given expression) and then running the FSM against the desired input stream
(in some cases the FSM computation is skipped and the matching is done against the regexp directly, but this distinction is not relevant here).
In this report we consider Mealy machines,
which are commonly used in the field of natural-language processing~\cite{Kaplan1994,Mohri2005}
and differ from DFA in that they react to the input read by producing (or not producing) output symbols.
With more than one output symbol, such as $\alpha$ and $\beta$,
we can discriminate amongst different patterns detected.
Next we will present two possible attempts at solving the present problem that could be carried out with standard regexp engines, as we know them, and the limitations they have.

As a first attempt at complete matching of $e_1$ and $e_2$, 
we could convert each regexp to its corresponding FSM and then run both of them, in parallel, 
against the same input.
However, there would be problems in both the correctness and the efficiency of this solution:
\begin{itemize}

 \item As already mentioned, 
 regexp engines in general will not detect matches that overlap.
 In particular, 
 the second match of $e_2$ in $s$ would be overlooked since it overlaps with the first one.
 
 \item We could achieve complete matching of $e_2$ based on its FSM by forcing the engine to repeat the search from every index in the input string, 
 including backtracking to an earlier index when a match ending in a later one has been found. 
 For example, in trace $s$, the engine would have to attempt a match 13 times (one for every symbol),
 including backtracking to the 4th index after the first match of $e_2$ (which lasts from the 3rd index to the 11th) has been processed.
 This would come with an efficiency cost, 
 as backtracking is often considered to be the major cause of inefficiency in regexp tools.
 
\item Backtracking restricted to the outside of matches would suffice when finding only non-overlapping matches  
(and this probably approximates the behaviour of most regexp engines), 
but in the case of real-time triggering,
 where each input symbol is to be read, ideally, only once,
this is still undesirable, if at all affordable.

\item Another aspect in which the efficiency of this attempt falls short is the fact that, in general,
computing FSMs from multiple regexps individually often produces more states (and never fewer) altogether 
than computing a single FSM where each regexp is a sub-pattern of a larger regexp.
If we write $e_1$ and $e_2$ as two sub-patterns of regexp $e_3$ (see Table~\ref{table:three_expressions}), 
we can then compute a single FSM $\M$ containing 8 states (see Figure~\ref{fig:dfst_e3_min})
as opposed to the 10 states in total from the two FSMs.
 
\end{itemize}

\begin{figure}[hbtp]
\centering
  \includegraphics[width=.5\linewidth]{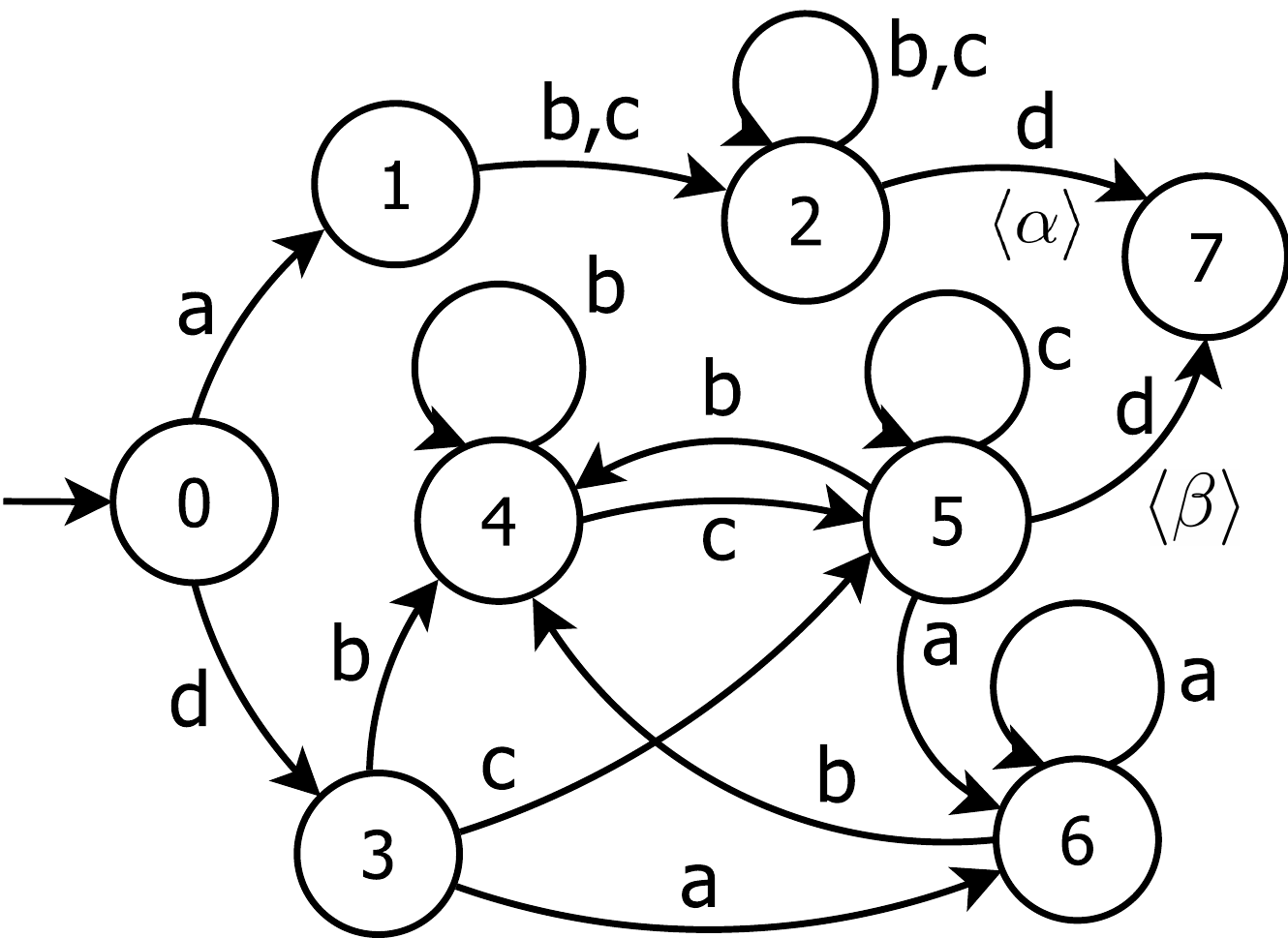}
\caption{A Mealy machine $\M$ detecting (exact) instances of $e_3$.}
\label{fig:dfst_e3_min}
\end{figure}

To accommodate the last point made,
we could re-attempt complete matching
but this time using just $\M$.
This, however, would create new problems:
\begin{itemize}

 \item Note that now even detecting non-overlapping instances of the sub-patterns given is not guaranteed to work:
in trace $s$, both instances of $e_2$ overlap with a previous instance of $e_1$ (and therefore of $e_3$),
making both of them go undetected by a common regexp engine.

 \item As above, this could be resolved by repeating the exploration from every index in the input,
 but this time the efficiency cost would be even greater,
 since each exploration would now be on a larger FSM. 
 
\end{itemize}

The unified machine $\M$ brings us closer to the solution we will present,
but as we have just seen it still does not allow us to find \emph{all} matches in arbitrary sequences while reading each symbol only once.
An FSM capable of doing so would need to, at least, have transitions for all symbols from every state
(i.e., be \emph{complete}).
However, as we will see,
there is no way of making $\M$ complete that could possibly solve the problem.

Finally, 
we are interested in working with FSMs in their smallest representations possible
(for instance, when constructed from regexps, FSMs in general are not minimal).
We will show a variation of Hopcroft's minimization algorithm for DFA that reduces Mealy machines to the fewest states possible.
We also consider a different minimization algorithm proposed for Mealy machines~\cite{Solovev2011}, 
which we show is not complete.

\paragraph{Report outline.} We start by visiting in Section~\ref{sec:preliminaries} the foundational concepts of regular expressions and languages, automata and transducers.
In Section~\ref{sec:patternRegexps} we define a grammar for writing regexps representing patterns with input and output symbols, 
as well as their equivalent representations in the form of non-deterministic transducers.
Section~\ref{sec:determinization} shows how to determinize transducers into Mealy machines,
and we show in Section~\ref{sec:minimization} how any Mealy machine can be minimized to its (unique) minimal equivalent machine by an indirect application of a DFA minimization algorithm. 
In Section~\ref{sec:completeMatching} we construct our theoretical model for complete regexp matching while processing input symbols only once.
Appendix~\ref{sec:app:stateMerges} contains a counter-example showing that the minimization algorithm from \cite{Solovev2011} does not apply all reductions possible.
Appendix~\ref{sec:app:automata_images} provides additional figures referenced in the report.

\section{Preliminaries}		\label{sec:preliminaries}

We briefly present some fundamental notions of regular expressions and languages,
automata and transducers.

Let the finite set $\AlphbPre$ be an \emph{alphabet}.
Any finite sequence of letters is called a \emph{word} over $\AlphbPre$ and the \emph{empty word} is denoted by $\varepsilon$.
Let $\AlphbPre^*$ be the set of all words over $\AlphbPre$.
A \emph{language} over $\AlphbPre$ is a subset of $\AlphbPre^*$.
The set of regular expressions (regexps) over $\AlphbPre$, $E_\AlphbPre$, is defined by 
$$ e \in E_\AlphbPre \; := \; \varepsilon \mid \letter \in \AlphbPre \mid e' \cdot e'' \mid e' + e'' \mid e'^*, \; \mbox{where $e',e'' \in E_\AlphbPre$.} $$  

The operator $\cdot$ (concatenation) is often omitted.
For every $e$, its language can be defined inductively as 
$\Lang(\varepsilon) = \{ \varepsilon \}$,
$\Lang(\letter) = \{\letter\}$ for $\letter \in \AlphbPre$,
$\Lang(e' + e'') = \Lang(e') \cup \Lang(e'')$,
$\Lang(e' \cdot e'') = \Lang(e') \cdot \Lang(e'')$ and
$\Lang(e'^*) = \Lang(e')^*$.

Regular languages can also be represented by automata.
A \emph{Nondeterministic Finite-Word Automaton} (NFA) is defined by a 5-tuple
$(Q,\Sigma,\delta,I,F)$, where $Q$ is a finite set of states,
$\Sigma$ an input alphabet, 
$\delta \subseteq Q \times (\Sigma \cup \{\varepsilon\}) \times Q$ a \emph{transition relation},
$I \subseteq Q$ a set of \emph{initial states} and $F \subseteq Q$ a set of \emph{final states}.
A word $w \in \Sigma^*$ is in the language of state $p$,
$\Lang(p)$,
 if and only if it can be read from $p$ and ending in a final state by taking transitions that follow the symbols in $w$.
Transitions labeled with $\varepsilon$ do not read any symbol and can always be taken.
The union of the languages of all initial states gives the language of the NFA.
For every expression $e \in E_\AlphbPre$, 
an NFA accepting the same language can be computed by applying Thompson's construction~\cite{Louden1997}.
%\footnote{Also known as the McNaughton-Yamada-Thompson algorithm~\cite{Aho2006}.}. 

A \emph{Deterministic Finite-Word Automaton} (DFA) is given by $(Q$,$\Sigma$,$\Delta$,$i$,$F)$,
where $Q$, $\Sigma$ and $F$ are defined as above, 
$\Delta: Q \times \Sigma \rightarrow Q$ is a (partial) transition function and $i$ is the initial state.
We treat DFA as \emph{complete} only if $\Delta$ is total.
Every NFA $(Q,\Sigma,\delta,I,F)$ can be converted to a DFA $(Q',\Sigma,\Delta,i,F')$ accepting the same language by applying the \texttt{subset} construction~\cite{Aho2006}.
%where $Q' \subseteq P(Q)$, $i = \varepsilon(I)$ and 
%$F' = \{ S \subseteq Q \mid \exists q \in F \cdot q \in S \}$.
The algorithm first sets $i$ to $\varepsilon(I)$
and then constructs $\Delta(i,\sigma) = \varepsilon(move(i,\sigma))$,
for all $\sigma \in \Sigma$,
where $move(S,\sigma)$ is the set of states that can be reached by $\sigma$ from states in $S$.
The same process is iterated over all new subsets of $Q$ until all subsets constructed have been processed.
At the end,
all subsets constructed are stored in $Q' \subseteq P(Q)$ and $F'$ is set to 
$\{ S \in Q' \mid F \cap S \neq \emptyset \}$.

As a consequence of the Myhill–Nerode theorem,
for every DFA there is a \emph{minimal} 
(with the fewest possible states) equivalent DFA and it is unique~\cite{Hopcroft1979}.
Computing the minimal DFA is known as \emph{minimization} and
different algorithms exist,
like Hopcroft's algorithm~\cite{Hopcroft1971,berstel2010},
which has the best worst-case complexity known
($O(kn \log n)$, for $n$ the number of states and $k$ the size of the alphabet),
and Moore's algorithm ($O(kn^2)$)\cite{church_1958}.
Both algorithms are based on partition refinement:
partitioning the states into equivalence classes based on their behaviour (the words they can read)
and merging all states in a class into a macrostate in the minimal automaton.

A \emph{Finite-State Transducer} (FST) is a nondeterministic machine with two memory tapes, 
an input tape and an output tape~\cite{berstel, elgot}. 
It can be seen as an NFA that reads and writes symbols.
Formally, an FST is a 6-tuple $(Q,\Sigma,\Gamma,I,F,\delta)$,
where $Q$, $\Sigma$, $I$ and $F$ are defined as in NFA,
$\Gamma$ is an \emph{output alphabet} and
$\delta \subseteq Q \times (\Sigma \cup \{\varepsilon\}) \times (\Gamma \cup \{\varepsilon\}) \times Q$ a transition relation.
A transition $(p,\sigma,\gamma,q) \in \delta$ 
reads \emph{from $p$ we may read $\sigma$ and move to $q$, 
in which case $\gamma$ is emitted},
and we call $p$ and $q$ the \emph{departing} and \emph{arriving} states, respectively.
Transitions may be by the empty word ($\sigma = \varepsilon$) or have no output ($\gamma = \varepsilon$).
%We write $p \transNO{\sigma} q$ when the output (if it exists) need not be specified,
In NFA and FST, 
the \emph{$\varepsilon$-closure} of $p$ (here denoted by $\varepsilon(p)$) is
the set of states reachable from $p$ by $\varepsilon$-transitions,
and we always have $p \transNO{\varepsilon} p$.
The closure operation can be lifted to sets of states.

A \emph{Mealy machine}~\cite{mealy,mohri}, also known as a deterministic FST, 
is a 6-tuple $(Q,\Sigma,\Lambda,i,T,G)$, 
where $Q$, $\Sigma$ and $i$ are defined as for DFA,
$\Lambda$ is an output alphabet, and
$T: Q \times \Sigma \longrightarrow Q$ and 
$G: Q \times \Sigma \longrightarrow \Lambda \cup \{\varepsilon\}$ are
transition functions that map a departing state and an input symbol to an arriving state and either the output symbol emitted or the empty word, respectively. %OR $\Gamma$ ?
%We use similar notations $p \trans{\sigma}{\lambda} q$ and $p \transNO{\sigma} q$ to the ones defined for FST. 
%Similarly to FST,
%we can write $p \trans{\sigma}{\lambda} q$ and $p \transNO{\sigma} q$.
%As with $\Delta$ in DFA,
We treat $T$ and $G$ as partial functions (with the same domain) and only if they are total do we say the Mealy machine is complete.
In FST and Mealy machines we sometimes write transitions as $p \trans{\sigma}{\gamma} q$.
When the output (if it exists) need not be specified, we can just write
$p \transNO{\sigma} q$, 
and also $p \stackrel{A}{\longrightarrow} q$ when $p \transNO{\sigma} q$ for every symbol $\sigma \in A \subseteq \Sigma$.
When $q$ is reachable by taking transitions that follow the symbols in $w \in \Sigma^*$ we write $p \transNO{w} q$ and
$p \transNOstar{w} q$ if before and after the symbol transitions there may be 1 or more $\varepsilon$-transitions.

\section{From Pattern Regexps to Complete Matching}

In this section we show how to convert any regexp $e$ with input and output symbols into the smallest Mealy machine
performing complete matching of $e$
%(i.e., finds all matches, including in cases of overlap like (1) and (2) as defined in Section~\ref{sec:introduction}) 
on any input sequence of arbitrary length.
As we will see, the regexp may encode any number of patterns (and sub-patterns inside a pattern) and the computed machine is able to discriminate amongst them upon a match.

\subsection{Pattern Regexps and Finite-State Transducers}	\label{sec:patternRegexps}

Let $\Sigma$ and $\Gamma$ be respectively input and output alphabets.
We define a \emph{pattern regexp} as a regexp in $E_\Alphb$,
where $\Alphb = \Sigma \times (\Gamma \cup \{ \varepsilon \})$
\emph{unifies} input and output symbols to naturally capture the notion that outputs are produced as reactions to the input read,
e.g., $(a,\alpha)$ (\emph{output $\alpha$ on reading $a$}),
but not all symbols read cause output symbols to be emitted, e.g., $(a,\varepsilon)$.
We represent pairs $(\sigma,\gamma)$ as $\sigma \langle \gamma \rangle$ and, for simplicity,
pairs with no output $(\sigma,\varepsilon)$ as simply $\sigma$.
We can see all expressions in Table~\ref{table:three_expressions} (Section~\ref{sec:introduction}) are pattern regexps.

While regexps are nondeterministic models,
the notion that output is produced as a reaction to the input read requires treating the output symbols deterministically.
This can be captured by a notion of \emph{behaviour} of a pattern regexp. 
While language defines the set of words that match a given regexp,
we define the behaviour of a pattern regexp as the function that maps input words after which an output is produced to the symbols they emit. 
This is formally defined based on the corresponding language and using the auxiliary notation $w|_{\Sigma}$ that
denotes the sequence of input symbols in a string $w \in \Alphb^*$,
such that $\varepsilon|_{\Sigma} = \varepsilon$ and $((\sigma,o)\cdot w')|_{\Sigma} = \sigma \cdot w'|_{\Sigma}$.

\begin{definition}	\label{def:behaviourExpression}
The \emph{behaviour of a regexp} $e \in E_\Alphb$ is the partial function
$\B(e) : \Sigma^* \rightarrow \Parts(\Gamma)$ where
$\B(e)(\varepsilon) = \{ \gamma \in \Gamma \mid (\varepsilon,\gamma) \in \Lang(e) \}$ and
$\forall_{w \in \Sigma^*} \forall_{\sigma \in \Sigma} \cdot \B(e)(w\sigma) = \{ \gamma \in \Gamma \mid \exists_{\alp \in \Alphb^*} \cdot 
\alp|_\Sigma = w \land \alp \!\cdot\! (\sigma,\gamma) \in \Lang(e) \}$.
\end{definition}

It is easy to see that expressions in $E_\Alphb$ may have different languages but the same behaviour
%(e.g.\ $a\langle \alpha\rangle$, $a+a\langle \alpha\rangle$ and $a\langle \alpha\rangle a^*$)
(e.g.\ $a\langle \alpha\rangle$ and $a+a\langle \alpha\rangle$)
but not vice-versa.
From a pattern regexp we can then compute a corresponding FST, 
by first computing an NFA via Thompson's construction and then
modifying all transitions of the form $\sigma \langle \gamma \rangle$ to $\trans{\sigma}{\gamma}$. 
The FST computed via this construction for expression $e_3$ from Section~\ref{sec:introduction} can be seen in Appendix~\ref{sec:app:automata_images} (Figure~\ref{app:fig:fst_e3}).

While in NFA and DFA a word $w$ is in the language of a state $p$ if a final state can be reached by reading $w$ from $p$,
in a transducer the \emph{behaviour} of $p$ maps $w$ to a set of symbols if they can be emitted after reading $w$ from $p$.
Definition~\ref{def:behaviourTransducerState} captures the notion that the symbols may be emitted by a transition reading the last symbol in $w$ or by any $\varepsilon$-transition that follows.

\begin{definition}	\label{def:behaviourTransducerState}
The \emph{behaviour of a state $p$} in a transducer is the partial function
$\B(p) : \Sigma^* \rightarrow \Parts(\Gamma)$, where 
$\B(p)(\varepsilon) = \{ \gamma \in \Gamma \mid \exists_{q,r \in \varepsilon(p)} \cdot q \trans{\varepsilon}{\gamma} r \} $
and $\forall_{w \in \Sigma^*} \forall_{\sigma \in \Sigma} \cdot \B(p)(w\sigma) = 
\{ \gamma \in \Gamma \mid \exists_{\gamma' \in \Gamma \cup \{\varepsilon\}} \exists_{q,r \in Q} \cdot p \transNOstar{w} q \trans{\sigma}{\gamma'} r \land ( \gamma = \gamma' \lor \gamma \in \B(r)(\varepsilon) ) \}$.
\end{definition}

Since the output produced defines how a transducer (or regexp) reacts or does not react to a given input, 
we make no distinction between transducers (or regexps) that differ only for inputs for which no output is produced.
Thus, we speak of the \emph{outputting behaviour} (or simply \emph{output}) of a state (or regexp) as the partial function
$\Bo(p) : \Sigma^* \rightarrow \Parts(\Gamma)$ such that $\Bo(p)(w)$ is defined and is equal to $\B(p)(w)$ if and only if $\B(p)(w) \neq \emptyset$.
Pattern regexps may then have different behaviours but the same output 
(e.g.\ $a\langle \alpha\rangle$ and $a\langle \alpha\rangle a^*$).

The \emph{output of a transducer} $\M$ can then be defined by that of its initial states:
$\forall_{w \in \Sigma^*} \cdot \Bo(\M)(w) = \{ \gamma \mid \exists_{i \in I} \cdot \gamma \in \Bo(i)(w) \}$.
In a Mealy machine, the mapping is into singleton sets of symbols in $\Lambda$.
We say two transducers $\M$ and $\M'$ are \emph{equivalent} if and only if
$\Bo(\M) = \Bo(\M')$.

\subsection{From Finite-State Transducers to Mealy Machines}	\label{sec:determinization}

In this section we consider the problem of determinizing transducers.
Note that FSTs that emit output symbols before any input is read do not have an equivalent deterministic transducer.
For all other FSTs, we obtain an equivalent Mealy machine by adapting the subset construction for NFA.

Given an FST $(Q,\Sigma,\Gamma,I,F,\delta)$, 
let us denote by $\texttt{subsetT}$ the algorithm that constructs the Mealy machine $(Q',\Sigma,\Lambda,i,T,G)$,
where $Q' \subseteq P(Q)$, $\Lambda \subseteq \Parts(\Gamma)$ and $i = \varepsilon(I)$.
The transition function $T$ maps a state $S \subseteq Q$ and an input symbol $\sigma \in \Sigma$ to 
$T(S,\sigma) = \varepsilon(\{ q \mid \exists_{p \in S} \cdot p \transNO{\sigma} q \})$,
the set of all states that can be reached from states in $S$ by a transition by $\sigma$ followed by 0 or more transitions by $\varepsilon$.
The transition function $G$ maps a state $S$ and a symbol $\sigma$ to the set $G(S,\sigma) = G_\sigma(S) \cup G_\varepsilon(T(S,\sigma))$,
where $G_\sigma(S) = \{ \gamma \in \Gamma \mid \exists_{q \in S, q' \in Q} \cdot q \trans{\sigma}{\gamma} q' \}$ contains all output symbols emitted by transitions by $\sigma$ from states in $S$,
and $G_\varepsilon(T(S,\sigma)) = \{ \gamma \in \Gamma \mid \exists_{p,q \in T(S,\sigma)} \cdot p \trans{\varepsilon}{\gamma} q\}$
contains all output symbols that are emitted without reading any input after $\sigma$ is read from $S$.
The $\texttt{subsetT}$ construction iterates over all subsets constructed as in the case of an NFA:
first it computes $i$ and all its outgoing transitions using $T$ (and $G$),
and then the same process is repeated for the new subsets of $Q$ until all subsets constructed have been processed.
Finally, $Q'$ contains $i$ and all subsets computed by $T$ during the construction,
and $\Lambda$ contains all non-empty sets of symbols computed by $G$, while the empty set is encoded as $\varepsilon$.
Figure~\ref{fig:dfst_e3} shows the Mealy machine for $e_3$ (Section~\ref{sec:introduction}) computed by \texttt{subsetT} from the FST in Appendix~\ref{sec:app:automata_images} (Figure~\ref{app:fig:fst_e3}).
Proposition~\ref{prop:subsetT} states the correctness of \texttt{subsetT}.

\begin{proposition}		\label{prop:subsetT}
 For any FST $\M = (Q,\Sigma,\Gamma,I,F,\delta)$, 
 $\Bo(\M) = \Bo(\M')$, for $\M' = (Q',\Sigma,\Lambda,i,T,G)$ the Mealy machine such that $\M' =$ \texttt{subsetT\,$(\M)$}.
\end{proposition}
\begin{proof}
 We will show $\B(M)=\B(M')$, which implies $\Bo(M)=\Bo(M')$.

 We start by lifting $\B$ to sets of states as follows: 
 $ \forall_{S \subseteq Q} \forall_{w \in \Sigma^*} \cdot \B(S)(w) = \cup_{s \in S} \B(s)(w)$.
 Let us now show that the behaviour of every subset of states computed by \texttt{subsetT} 
 (i.e., every state in the resulting Mealy machine)
 is the behaviour of all FST states in the subset taken together.
 We will first show that the behaviour is preserved for words of one symbol, by showing
 1) $\forall_{S \subseteq Q} \forall_{\sigma,\sigma' \in \Sigma} \cdot \B(T(S,\sigma))(\sigma') = 
 \cup_{s \in S} \cup_{p \in \{ r \mid s \transNO{\sigma} r \}} \B(p)(\sigma')$.
 By the \texttt{subsetT} construction,
 the behaviour of $T(S,\sigma)$ is given by function $G$, such that
 $\B(T(S,\sigma))(\sigma') = G_{\sigma'}(T(S,\sigma)) \cup G_\varepsilon(T(T(S,\sigma),\sigma'))$.  
%$ \{ \gamma \in \Gamma \mid \exists_{q \in S, q' \in Q} \cdot q \trans{\sigma}{\gamma} q' \} \cup
% \{ \gamma \in \Gamma \mid \exists_{p,q \in T(S,\sigma)} \cdot p \trans{\varepsilon}{\gamma} q\}$.
 By the definition of $G$, we have 2)
 $G_{\sigma'}(T(S,\sigma)) = \{ \gamma \in \Gamma \mid \exists_{q \in T(S,\sigma)} \exists_{r \in Q} \cdot q \trans{\sigma'}{\gamma} r \}$, which, by the definition of $T(S,\sigma)$, can be rewritten as
 $\{ \gamma \in \Gamma \mid \exists_{s \in S} \exists_{p,q,r \in Q} \cdot 
 s \transNO{\sigma} p \transNOstar{\varepsilon} q \trans{\sigma'}{\gamma} r \} = $ 
 $\cup_{s \in S} \cup_{p \in \{ r \mid s \transNO{\sigma} r \}} 
 \{ \gamma \in \Gamma \mid \exists_{q,r \in Q} \cdot p \transNOstar{\varepsilon} q \trans{\sigma'}{\gamma} r \} $. 
 And knowing that $T(T(S,\sigma),\sigma') = T( \varepsilon(\{ q' \mid \exists_{p' \in S} \cdot p' \transNO{\sigma} q' \}) ,\sigma') =$
 $\varepsilon(\{ q'' \mid \exists_{p'' \in \varepsilon(\{q' \mid \exists_{p' \in S} \cdot p' \transNO{\sigma} q' \})} \cdot p'' \transNO{\sigma'} q'' \})$, we can show that 3) $G_\varepsilon(T(T(S,\sigma),\sigma')) = 
 \{ \gamma \in \Gamma \mid \exists_{p,q \in T(T(S,\sigma),\sigma')} \cdot p \trans{\varepsilon}{\gamma} q \} = $
 $\{ \gamma \in \Gamma \mid \exists_{p' \in S} \exists_{p,q,q',p'' \in Q} \cdot 
 p' \transNO{\sigma} q' \transNOstar{\varepsilon} p'' \transNO{\sigma'} q'' \transNOstar{\varepsilon} p \trans{\varepsilon}{\gamma} q \} =$
 $ \cup_{p' \in S} \cup_{q' \in \{ r \mid p' \transNO{\sigma} r \}} 
 \{ \gamma \in \Gamma \mid \exists_{p'',q'' \in Q} \cdot q' \transNOstar{\varepsilon} p'' \transNO{\sigma'} q'' \land \gamma \in \B(q'')(\varepsilon) \}$,
 from which we can obtain, by variable renaming,
 $ \cup_{s \in S} \cup_{p \in \{ r \mid s \transNO{\sigma} r \}} 
 \{ \gamma \in \Gamma \mid \exists_{q,r \in Q} \cdot p \transNOstar{\varepsilon} q \transNO{\sigma'} r \land \gamma \in \B(r)(\varepsilon) \}$.
 We finalize the proof of 1) as follows:
 $\cup_{s \in S} \cup_{p \in \{ r \mid s \transNO{\sigma} r \}} \B(p)(\sigma')$, by Definition~\ref{def:behaviourTransducerState}, 
 is equivalent to
 $\cup_{s \in S} \cup_{p \in \{ r \mid s \transNO{\sigma} r \}} \{ \gamma \in \Gamma \mid \exists_{\gamma' \in \Gamma \cup \{\varepsilon\}} \exists_{q,r \in Q} \cdot p \transNOstar{\varepsilon} q \trans{\sigma'}{\gamma'} r \land (\gamma = \gamma' \lor \gamma \in \B(r)(\varepsilon)) \} =$
 $\cup_{s \in S} \cup_{p \in \{ r \mid s \transNO{\sigma} r \}} (\{ \gamma \in \Gamma \mid \exists_{q,r \in Q} \cdot p \transNOstar{\varepsilon} q \trans{\sigma'}{\gamma} r \} \cup \{ \gamma \in \Gamma \mid \exists_{q,r \in Q} \cdot p \transNOstar{\varepsilon} q \transNO{\sigma'} r \land \gamma \in \B(r)(\varepsilon) \}) =$
 $(\cup_{s \in S} \cup_{p \in \{ r \mid s \transNO{\sigma} r \}} \{ \gamma \in \Gamma \mid \exists_{q,r \in Q} \cdot p \transNOstar{\varepsilon} q \trans{\sigma'}{\gamma} r \} )\cup (\cup_{s \in S} \cup_{p \in \{ r \mid s \transNO{\sigma} r \}} \{ \gamma \in \Gamma \mid \exists_{q,r \in Q} \cdot p \transNOstar{\varepsilon} q \transNO{\sigma'} r \land \gamma \in \B(r)(\varepsilon) \})$, from which, applying 2) and 3), we obtain
 $G_{\sigma'}(T(S,\sigma)) \cup G_\varepsilon(T(T(S,\sigma),\sigma')) = B(T(S,\sigma))(\sigma')$.

 Since, for a given $w \in \Sigma^+$, each symbol is read from a state,
 from 1) it follows 4) $\forall_{S \subseteq Q} \forall_{\sigma \in \Sigma} \forall_{w \in \Sigma^+} \cdot \B(T(S,\sigma))(w) = 
 \cup_{s \in S} \cup_{p \in \{ r \mid s \transNO{\sigma} r \}} \B(p)(w)$.
 A proof along the same lines as the proof above can be built for 
 5) $\forall_{\sigma \in \Sigma} \cdot \B(i)(\sigma) = \cup_{s \in i} \B(s)(\sigma)$,
 and since $\M$ does not produce outputs before any input is read, 
 we have 6) $\B(i)(\varepsilon) = \B(I)(\varepsilon) = \emptyset$.
 From 4), 5) and 6) we obtain $\forall_{w \in \Sigma^*} \cdot \B(i)(w) = \cup_{s \in I} \B(s)(w)$,
 and consequently $\B(M') = \B(M)$.
\end{proof}

\begin{figure}[hbtp]
\centering
  \includegraphics[width=.6\linewidth]{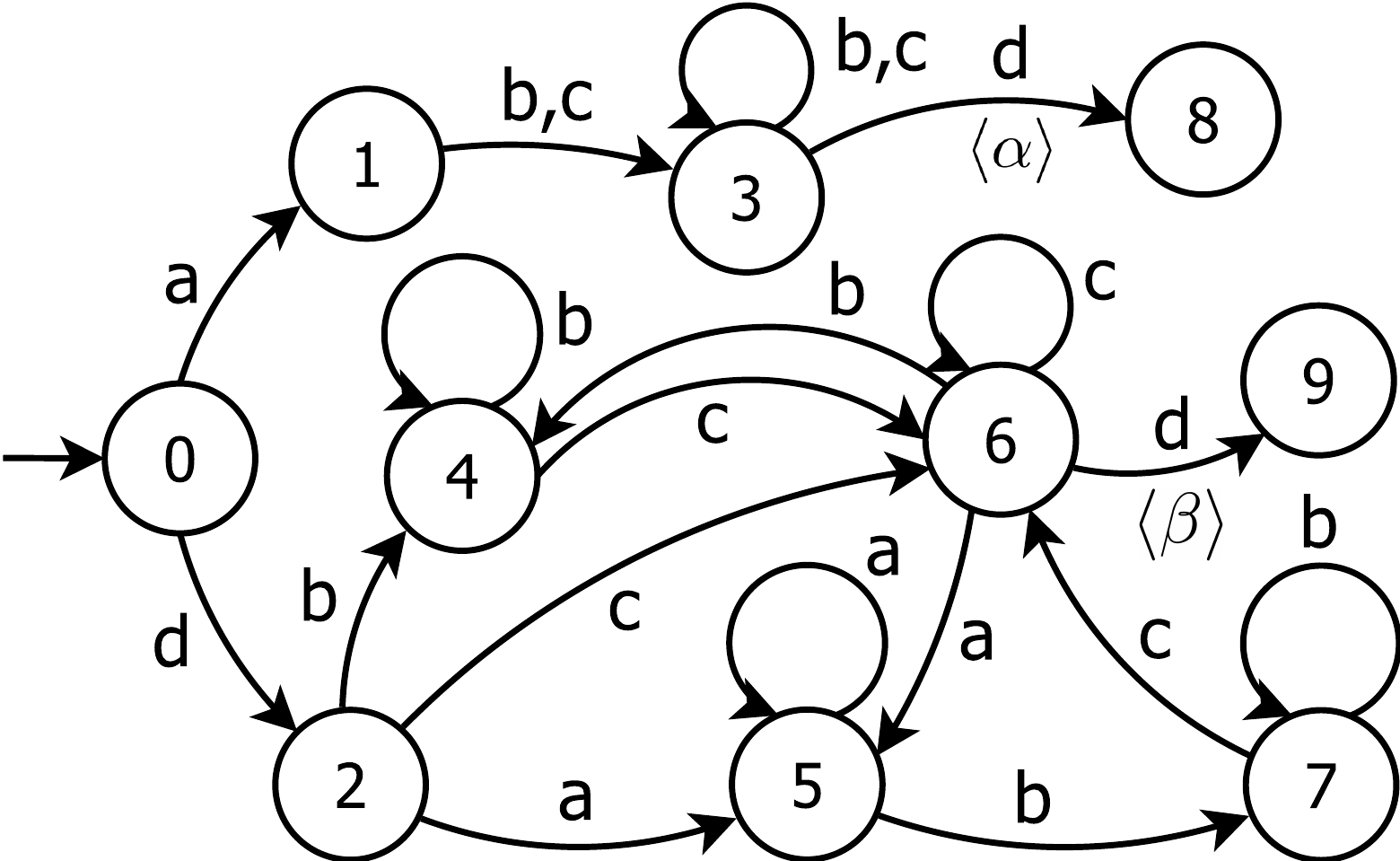}
\caption{A Mealy machine for 
  $e_3 = a(b+c)^+d\langle\alpha\rangle + d((a^*b^++b^*)c)^+d \langle\beta\rangle$.}
\label{fig:dfst_e3}
\end{figure}

As in the case of DFA, 
the Mealy machines computed by the subset construction, in general, are not minimal.
In the next section we address the problem of minimization.

\subsection{Minimizing Mealy Machines}	\label{sec:minimization}

In this section we consider the problem of reducing a Mealy machine to the smallest possible number of states while retaining its behaviour,
also known as the minimization of a Mealy machine.
We will consider different possible approaches and their disadvantages,
and we will present one, based on indirect DFA minimization,
that computes a Mealy machine that is at the same time minimal and complete.
When using Hopcroft's algorithm for DFA minimization, 
the solution here presented becomes the most efficient one to the best of our knowledge,
and we provide a worst-case complexity analysis.

The subject of minimization has been more extensively studied for DFA than for Mealy machines.
A well-known algorithm for minimizing DFA is Moore's algorithm,
originally defined for Moore machines with one output symbol.
Moore machines are state machines where the outputs are determined by the current state,
as opposed to the transitions taken as in Mealy machines.
Moore and Mealy machines are equally expressive and it is not difficult to translate one into the other.
It would thus be possible to convert a Mealy machine to a Moore machine, 
apply an extension of Moore's algorithm that deals with any number of output symbols
and translate it back to a Mealy machine, 
but this would come with two disadvantages.
First, Moore machines are, in general, larger than Mealy machines:
every state in a Mealy machine would result in one or more new states in a Moore machine,
one for each different output (or no output) in the incoming transitions of that state
(for example, the Mealy machine with 9 states in Figure~\ref{fig:rest_Mealy machine_e3_min} would would require 3 new states,
and 12 new transitions, as a Moore machine).
Second,
Moore's algorithm performs worse in the worst case ($O(kn^2)$, for $n$ states and $k$ input symbols~\cite{church_1958}) 
than other DFA minimization algorithms known, 
such as Hopcroft's algorithm, and so better alternatives exist.

A minimization algorithm for Mealy machines has been presented in \cite{Solovev2011},
 applying different state-merging techniques 
(also followed in \cite{Klimovich2012} and more recently in \cite{Klimowicz2020}).
However, we present a counter-example (see Appendix~\ref{sec:app:stateMerges}) showing that, 
contrary to the claim, these techniques are not complete. 

Below, we present a solution for minimizing a Mealy machines to its unique, equivalent machine that is minimal and complete:

\vspace*{4pt}

\noindent
\minimize($\M = Q,\Sigma,\Lambda,i,T,G$):
\begin{enumerate}
 \item Make $\M$ complete by adding an absorbing state $\bot$ to $Q$ and transitions
 $T(p,\sigma) = (\varepsilon, \bot)$ for all $\sigma \in \Sigma$ and all $p \in Q$ with no transition by $\sigma$.\label{minimize:completeSigma}
 \item Let $\A = (Q,\Alphb,\Delta,i,F)$, with $F=Q$, be the DFA translated from $\M$,
 		where $U = \Sigma \times (\Lambda \cup \{\varepsilon\})$ and $\Delta$ joins $T$ with $G$. 	 \label{minimize:considerA}
 \item Make $\A$ complete by adding a non-final state $\bot'$ and transitions
  $\Delta(p,\alp) = \bot'$ for all $\alp \in \Alphb$ and all $p$ with no transition by $\alp$.\label{minimize:completeSigmaLambda}
 \item Compute the minimal DFA $\A' = (Q',\Alphb,\Delta',i',F')$ of $\A$.\label{minimize:minimizeA}
 \item Remove from $\A'$ the subset with $\bot'$ and all its incoming transitions.\label{minimize:removeBot}
 \item Return the machine $\M' = (Q',\Sigma,\Lambda,i',T,G)$ translated from $\A'$.\label{minimize:return}
\end{enumerate}

An implementation can perform all steps in the input machine (translations are only implicit).

\begin{proposition}		\label{prop:correctnessMinimize}
For any Mealy machine $\M$, 
$\minimize(\M)$ computes the unique minimal complete Mealy machine $\M'$ such that $\Bo(\M) = \Bo(\M')$.
\end{proposition}
\begin{proof}
 We start by defining the auxiliary function 
 $\Out(p,s)$ that gives the output symbol 
 $o \in \Lambda \cup \{\varepsilon\}$ emitted by the Mealy machine after reading the input string $s \in \Sigma^*$ from state $p$.

 In \ref{minimize:completeSigma}, we start by adding input-only transitions to $\M$,
 which preserves the behaviour and, in \ref{minimize:considerA}, we translate it to the DFA $\A$.
 Let us show that 
 (i) for every mapping in $\B(p)$ there corresponds one string in $\Lang(p)$ and (ii) vice-versa.
 To show i), since $\M$ is complete,
 we know that for any $p \in Q$ in $\M$ and any non-empty $s \in \Sigma^*$,
 $\Out(p,s)$ is defined and is equal to some $o \in \Lambda \cup \{\varepsilon\}$.
 Let $s'$ and $\sigma$ be such that $s = s' \cdot \sigma$.
Since $\A$ is a translation of $\M$, 
it then follows that there is some $w \in \Alphb^*$, such that $w|_{\Sigma} = s'$,
 and some $q,r \in Q$ such that $p \transNO{w} q \transNO{(\sigma,o)} r$ in $\A$,
 and $w$ is unique since $\M$ is deterministic.
 Since all states are final in $\A$, $w \cdot (\sigma,o) \in L(p)$, which shows (i).
 Conversely, for any $w \in \Alphb^*$
 and any $\sigma \in \Sigma$ such that $w \cdot (\sigma,o) \in \Lang(p)$, 
 for some $o \in \Lambda \cup \{\varepsilon\}$, 
 there exists one $s = w|_{\Sigma} \cdot \sigma$ such that $\Out(p,s) = o$, which shows (ii).

In \ref{minimize:completeSigmaLambda},
$\A$ is made complete,
a language-preserving step required before applying a DFA minimization algorithm.
In \ref{minimize:minimizeA},
$\A$ is minimized, 
which yields $\A'$.
In \ref{minimize:removeBot}, we remove the macrostate containing $\bot'$ and all its incoming transitions,
which preserves the language of $\A'$ since from $\bot'$ no final state is reachable.
In \ref{minimize:return},
$\A'$ is translated to Mealy machine $\M'$.
Next, we will show that $\M'$ is the minimal complete Mealy machine such that $\Bo(\M') = \Bo(\M)$ by showing that the states in $\M'$ cannot be merged any further.

Let us assume that $\M'$ is not the smallest possible complete Mealy machine for $\Bo(\M)$ and derive an absurd.
By hypothesis,
there are two different states $p$ and $q$ in $\M'$ such that $\Bo(p) = \Bo(q)$.
Since $p$ and $q$ have not been merged together in \ref{minimize:minimizeA},
we have $\Lang(p) \neq \Lang(q)$, and by (i) and (ii) this 
implies $\B(p) \neq \B(q)$.
Then, there is some $w \in \Sigma^*$ and some $\sigma \in \Sigma$ such that 
(iii) $w \cdot \sigma$ is a differentiator between $\B(p)$ and $\B(q)$ but $w$ is not.
Since $\M'$ is complete,
we know there are $p'$, $p''$, $q'$ and $q''$ such that 
$p \transNO{w} p' \trans{\sigma}{o} p''$ and $q \transNO{w} q' \trans{\sigma}{o'} p''$,
for some $o$ and $o'$ in $\Gamma \cup \{\varepsilon\}$.
From (iii), we have $o \neq o'$, which results in an absurd as it contradicts the hypothesis $\Bo(p) = \Bo(q)$. 

Therefore, we conclude that $\M'$ is the smallest possible complete Mealy machine for $\Bo(\M)$. 
The uniqueness comes from the fact that the minimization performed in step \ref{minimize:minimizeA} produces the unique minimal machine for $\Lang(\A)$ (or for $\B(\M)$) and from the fact that $\M'$ is complete.
\end{proof}

The Mealy machine with 10 states for $e_3$ (Figure~\ref{fig:dfst_e3}) can be minimized with the present algorithm to obtain the minimal Mealy machine with 8 states 
(see Figure~\ref{fig:dfst_e3_min} from Section~\ref{sec:introduction}).

\subsection{Achieving Complete Regexp Matching}	\label{sec:completeMatching}

In this section we show how to, 
given a pattern regexp,
build a Mealy machine that finds \textbf{all} matches (including those that overlap)
in any input sequence of arbitrary length.
We say such a machine performs \emph{complete regexp matching}.
Additionally, the given regexp may encode multiple patterns and sub-patterns and,
as we will see, the computed machine is able to signal, upon a match,
which of the patterns has been matched by producing an identifying output.
We will start by giving a formal definition of complete regexp matching and 
then we will present a solution for the example we saw in Section~\ref{sec:introduction},
which we then build up to obtain a general solution.

Let us note that, given a pattern expression and an input string,
detecting all matches of the expression in the string (including overlaps) can be formulated
as being able to find matches of the expression starting from any position in the string.
Therefore, and as captured in Definition~\ref{def:completeMatchingMachine},
an arbitrary input string should signal a match, when processed by a Mealy machine, 
for all its suffixes that correspond to matches of the expression.

\begin{definition}	\label{def:completeMatchingMachine}
Given alphabets $\Sigma$, $\Gamma$ and $\Alphb := \Sigma \times \Parts(\Gamma)$, a regexp $e \in E_\Alphb$ and a Mealy machine $\M$,
$\M$ performs \emph{complete regexp matching of $e$} if and only if 
$\forall_{w \in \Sigma^*} \cdot \B(M)(w) = 
\{ \gamma \in \Gamma \mid \exists_{w',w'' \in \Sigma^*} \cdot w=w'w'' \land \gamma \in \B(e)(w'') \}$.
\end{definition}

Recall the pattern expression $e_3$ from Section~\ref{sec:introduction} and the example matching trace \setuldepth{\large{$\langle\alpha\beta\rangle$}}
$s = \overline{ab\mbox{\ul{$d$}}}\underline{\langle\alpha\rangle bc\underline{abcbc\overline{d\raisebox{2.89mm}{}}}}\overline{\langle\alpha,\beta\rangle cd} \langle\beta\rangle$ given.
The Mealy machine in Figure~\ref{fig:dfst_e3_min} is the minimal machine for $e_3$,
matching exact instances of either $e_1$ or $e_2$,
and for that reason it does not perform complete regexp matching of $e_3$.
A necessary condition needed for complete matching is the ability to read arbitrary input sequences.
Thus, let us start by making the machine complete,
by adding all non-existing transitions to state 0 in an attempt to restart the regexp matching whenever there is a fail.
We denote the resulting machine by $\M$.

We will now give an illustration of how we can perform complete regexp matching of $e_3$ by having at our disposal multiple copies of $\M$ that can run in parallel. 
Consider we have $\M_1$, $\M_2$ and $\M_3$, 
three copies of $\M$ which are all at the initial state (0) before any input is read. 
Table~\ref{table:parallel_machines} summarizes how each machine will behave for each symbol of $s$ read. 
The first symbol is read by a machine at our choice (in this case $\M_1$) while all other machines stay at 0. 
After reading the first three symbols ($abc$), 
$\M_1$ is left at state 7 after producing the output symbol $\alpha$, 
as defined in $e_3$. 
We then read $d$, which makes $\M_1$ return to 0. 
However, $d$ itself is a possible prefix of instances of $e_3$,
and so we capture this by starting $\M_2$,
which is left at state 3 after reading $d$.
Since $\M_2$ would not have read the prefix $ab$ any differently than $\M_1$ did, 
it did not need be started until now. 
Moving on, reading $b$ causes $\M_2$ to move to state 5 while $\M_1$ stays at 0,
and then $ca$ leaves $\M_2$ at 6 and $\M_1$ at 1. 
Then, on reading $bcbc$, $\M_1$ ends up at 2 and $\M_2$ at 4. 
Looking ahead, and since at this point neither $\M_1$ nor $\M_2$ are at state 0, 
none of them will capture the following $d$ as a possible prefix of a new instance. 
Therefore, we make use of $\M_3$. 
As before, 
since $\M_3$ could only have read the symbols so far as either $\M_1$ or $\M_2$ did, we did not need to start it until now. 
And so, on reading $d$, $\M_3$ moves to 3, 
$\M_1$ moves to 7 while emitting $\alpha$, 
and $\M_2$ also moves to 7 but emitting $\beta$. 
At this point
there are two machines at the same state, 7, from where they could only duplicate each other for the remaining of the input,
and so we move one of them (we choose $\M_2$) to 0.
Finally, on reading $cd$, 
$\M_1$ ends at state 3 and $\M_3$ arrives at state 7 while producing $\beta$.

\begin{table}[hbtp]
\footnotesize
\begin{center}
\begin{tabular}{|c| c c c c c c c c c c c c c |}
\hline
       & $a$ & $b$ & $d$       & $b$ & $c$ & $a$ & $b$ & $c$ & $b$ & $c$ & $d$       & $c$ & $d$\\
\hline
$\M_1$ & 1   & 2   & 7$\langle\alpha\rangle$ & 0   & 0   & 1   & 2   & 2   & 2   & 2   & 7$\langle\alpha\rangle$ & 0   & 3 \\
\hline
$\M_2$ & 0   & 0   & 3                       & 5   & 4   & 6   & 5   & 4   & 5   & 4   & 0$\langle\beta\rangle$ &  0  &  0 \\
\hline
$\M_3$ & 0   & 0   & 0                       & 0   & 0   & 0   & 0   & 0   & 0   & 0   & 3      & 4    & 7$\langle\beta\rangle$ \\
\hline
\end{tabular}
\end{center}
\caption{Three copies of the machine $\M$ for $e_3$ working in parallel achieve complete regexp matching.}
\label{table:parallel_machines}
\end{table}

The correctness of the mechanism we illustrated for complete regexp matching of $e_3$ in $s$ can be argued in two main points.
On one hand,
since every output produced resulted from following the transitions in one of the copies of $\M$,
starting from its initial state,
we can know that all matches signalled are true positives.
On the other hand, 
since every symbol read caused all machines (on different states) to move and since, at every moment,
there was at least one machine at state 0 ready to start moving to capture the next match,
we know no match has been missed (i.e., there were no false negatives).
This solution can be generalized to any expression (and any input sequence) by simply making sure that, 
with every symbol read that causes all existing machine copies to not be in the initial state, 
a new copy of the machine is added and is set to its initial state.

We will now abstract the same reasoning presented above to compute a solution that performs complete regexp matching in a single Mealy machine.
As we saw, it is meaningless for the outcome produced by the mechanism to have copies of the machine at the same state,
for which reason, at each step, the number of machines used may be pruned to keep only machines at distinct states.
Moreover, the order of the states the machines are at does not matter 
(e.g., $\M_1$ being at state 1 and $\M_2$ at 0 has the same effect as $\M_2$ being at 1 and $\M_1$ at 0).
Therefore, at each step, the states of all machines together can be represented by a macrostate of $\M$.
The trace constructed in Table~\ref{table:parallel_machines} can then be written as transitions between macrostates:
$\{0\} \transNO{a} \{0,1\} \transNO{b} \{0,2\} \trans{d}{\alpha} \{0,3,7\} \transNO{b} \{0,5\} \transNO{c} \{0,4\} \transNO{a} \{0,1,6\} \transNO{b} \{0,2,5\} \transNO{c} \{0,2,4\} \transNO{b} \{0,2,5\} \transNO{c} \{0,2,4\} \trans{d}{\alpha,\beta} \{0,3,7\} \transNO{c} \{0,4\} \trans{d}{\beta} \{0,3,7\}$.
We can now see our reasoning effectively follows a variation of the subset construction on $\M$
where state 0 is added to every macrostate constructed.
Computing a macrostate from $\{0\}$, for every alphabet symbol, and then repeating the same construction over every macrostate constructed until all macrostates constructed have been processed, as is done by \texttt{subsetT}, 
results in a single machine that performs complete pattern regexp of $e_3$.
Generalizing to any Mealy machine, we obtain an algorithm for complete regexp matching for any pattern regexp.

Finally, 
note that since the initial state is added to every macrostate created, 
the initial machine ($\M$ in our example) need not be complete.
Moreover, since the solution extends the subset construction, 
the machine need not be deterministic either,
and thus the computation can be done in a single step from an FST. 
The algorithm then becomes \texttt{subsetTC},
which differs from algorithm \texttt{subsetT} from Section~\ref{sec:determinization} in that 
$T(S,\sigma)$ is replaced by
$T'(S,\sigma) = 
\varepsilon(\{ q \mid \exists_{p \in S} \cdot p \transNO{\sigma} q \}) \cup \{ i \} $.
The Mealy machine performing complete regexp matching of $e_3$,
computed with \texttt{subsetTC} from the FST of $e_3$ (Appendix~\ref{sec:app:automata_images}, Figure~\ref{app:fig:fst_e3})   and minimized with algorithm \texttt{minComp} (Section~\ref{sec:minimization}),
can be seen in Figure~\ref{fig:rest_Mealy machine_e3_min}.

\begin{figure}[hbtp]
\centering
  \includegraphics[width=.8\linewidth]{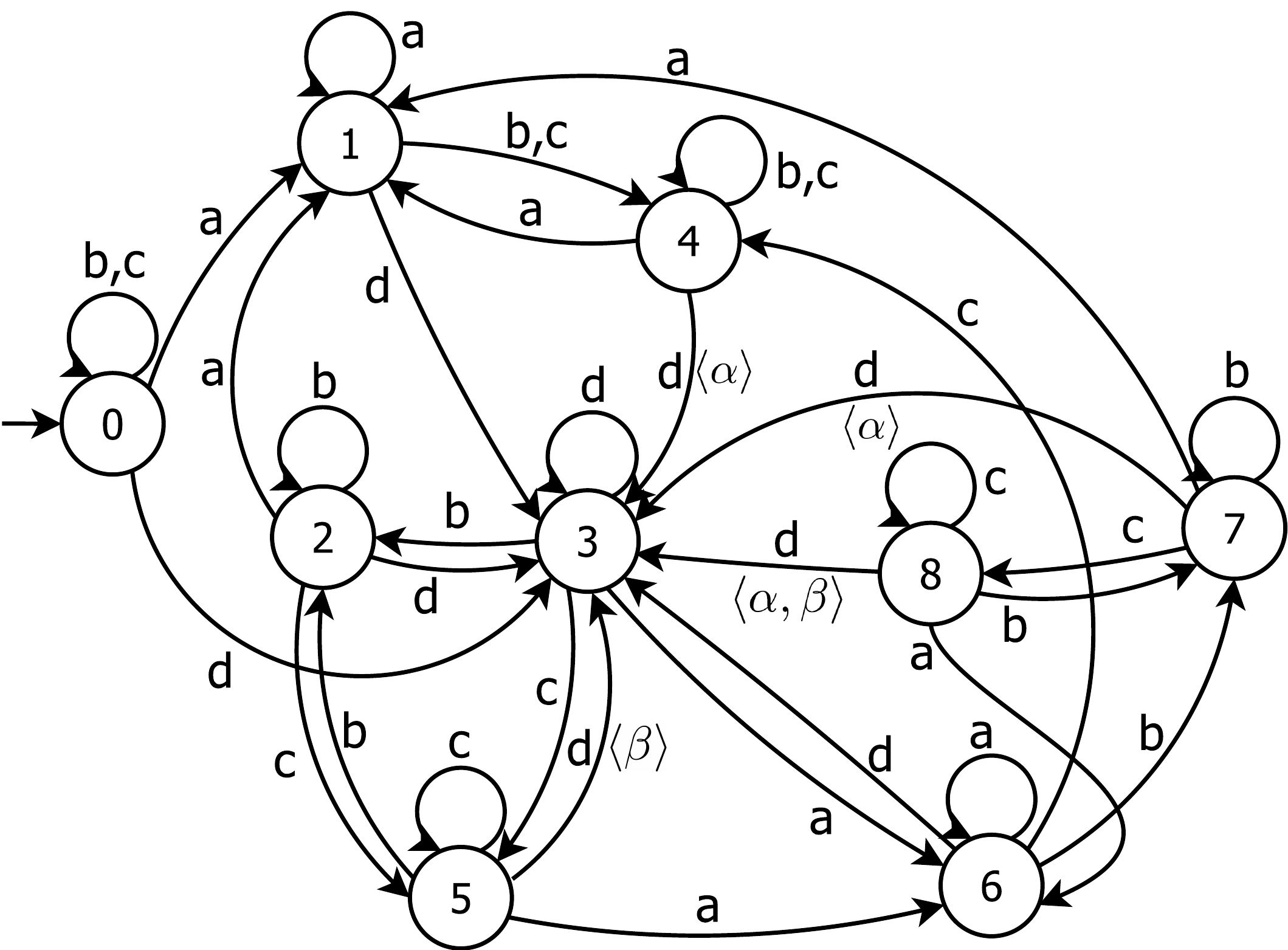}
\caption{The minimal Mealy machine performing complete regexp matching of
  $e_3 = a(b+c)^+d\langle\alpha\rangle + d((a^*b^++b^*)c)^+d \langle\beta\rangle$.}
\label{fig:rest_Mealy machine_e3_min}
\end{figure}

\begin{paragraph}{A note on languages over $\Sigma$.}
In the case of languages over a finite input alphabet 
$\Sigma = \{ \sigma_1, \ldots, \sigma_n \}$,
we can speak of complete matching of a regexp $e \in E_\Sigma$ by a DFA $\A$ by adapting Definition~\ref{def:completeMatchingMachine} to the language of $e$:
$
\forall_{w \in \Sigma^*} \cdot (w \in \Lang(\A) \Leftrightarrow \exists_{w',w'' \in \Sigma^*} \cdot w = w'w'' \land w'' \in \Lang(e))
$, or, in other words,
$\Lang(\A) = \{ w'w'' \in \Sigma^* \mid w'' \in \Lang(e) \} = \Sigma^* \cdot \Lang(e)$.
Therefore, if we denote by $NFA(e)$ an arbitrary NFA accepting $\Lang(e)$,
then computing \texttt{subsetTC}($NFA(e)$) becomes tantamount to computing 
\texttt{subset}($NFA(e')$) where 
$e' = (\sigma_1 + \ldots + \sigma_n)^* \cdot e$.
Indeed, the $NFA(e')$ in Figure~\ref{fig:nfa_input_only_languages} illustrates the idea introduced above of preserving the initial state throughout the subset construction:
since 0 contains a transition to itself by any symbol,
it will be present in every macrostate constructed.

\begin{figure}[hbtp]
\centering
  \includegraphics[width=.3\linewidth]{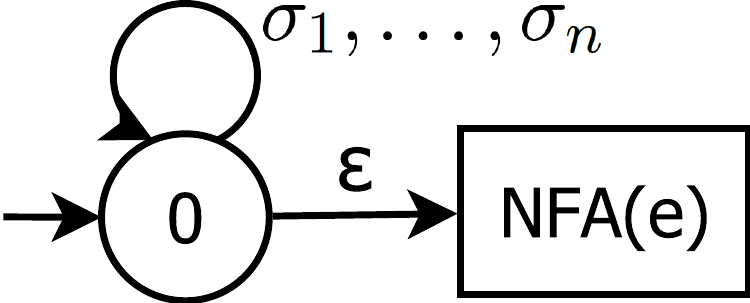}
\caption{For an input-only expression $e$ over $\Sigma = \{\sigma_1, \ldots, \sigma_n\}$,
computing \texttt{subsetTC}($NFA(e)$) is equivalent to computing \texttt{subset} on the NFA shown.}
\label{fig:nfa_input_only_languages}
\end{figure}

\end{paragraph}

\newpage

\appendix
\section{Appendix A:}	\label{sec:app:stateMerges}

In this appendix we provide a counter-example to the state-merging techniques presented in~\cite{Solovev2011} for minimizing a Mealy machine.   
The authors present a set of conditions as necessary and sufficient for merging two states in a Mealy machine while preserving its behaviour.
As we will show,
when applied to the machine in Figure~\ref{fig:minimization_counterexample_A}, 
the techniques fail to merge any states (thus making the rules not necessary).

The authors consider the more general case of Mealy machines where transitions may be initiated by multiple input symbols that must be read at the same time,
and by possibly more than one combination of symbols. 
For a transition from state $a_i$ to state $a_j$, 
the authors encode the transition rule as a ternary vector $X(a_i,a_j)$ of input variables initiating the given transition. 
Each position in this vector corresponds to a symbol in the alphabet and it will contain '$1$' if the transition can only be taken when the corresponding symbol is read, 
'$0$' if it can only be taken when the symbol is not read or '$-$' if the transition can be taken whether that symbol is read or not. 
Thus, the Mealy machines we consider in the our report are the particular case where every transition rule contains $1$ in the position of the one symbol initiating that transition while all other positions contain $0$. 
For example, 
in the Mealy machine with just one symbol shown in Figure~\ref{fig:minimization_counterexample_A}, $X(p,q) = X(q,r) = X(r,p) = 1$.
Moreover, two transitions such that, in some position, 
one contains $0$ and the other $1$ are said to be orthogonal. 
For a given transition rule $X(a_i,a_j)$, 
the set of output symbols emitted is denoted by $Y(a_i,a_j)$.

The authors represent the set of states corresponding to transitions from $a_i$ as $A(a_i)$, 
and divide the cases in which $a_i$ and $a_j$ can be merged into three possible situations:
$i) \; A(a_i) \cap A(a_j) = \{\}$, $ii) \; A(a_i) = A(a_j)$, 
and $iii) \; A(a_i) \neq A(a_j)$ and $A(a_i) \cap A(a_j) \neq \{\}$. 

For $i)$, 
the condition presented as necessary and sufficient to merge $a_i$ and $a_j$ is that for every $X(a_i,a_h)$ and every $X(a_j,a_t)$, 
where $a_h$ and $a_t$ are arbitrary states in the machine, 
$X(a_i,a_h)$ and $X(a_j,a_t)$ must be orthogonal. 
In the machine in Figure~\ref{fig:minimization_counterexample_A} we can see that,
while condition $i)$ is satisfied by every two states, 
all three transition rules are the same and thus not orthogonal,
and so no merging would occur.
However, 
it is easy to see there is an equivalent minimal machine with just one state and a transition to itself,
which is obtained by merging the three states into one since they have the same behaviour.
Therefore, the orthogonality condition given for situation $i)$ is not a necessary condition.

\begin{figure}[hbtp]
	  \centering
	  \includegraphics[width=.3\linewidth]{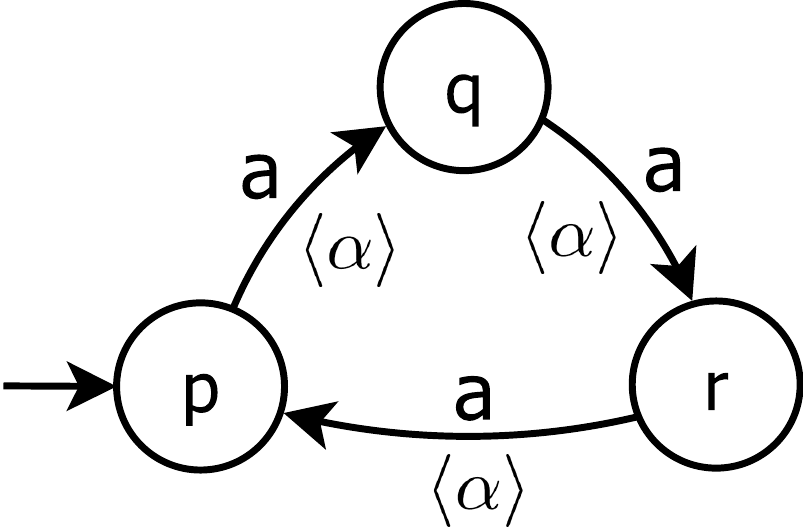}
	  \caption{All states are indistinguishable and therefore should be merged.}
	  \label{fig:minimization_counterexample_A}
\end{figure}

\newpage

\section{Appendix B:}	\label{sec:app:automata_images}

\begin{figure}[hbtp]
 \begin{center}
  \includegraphics[width=12cm]{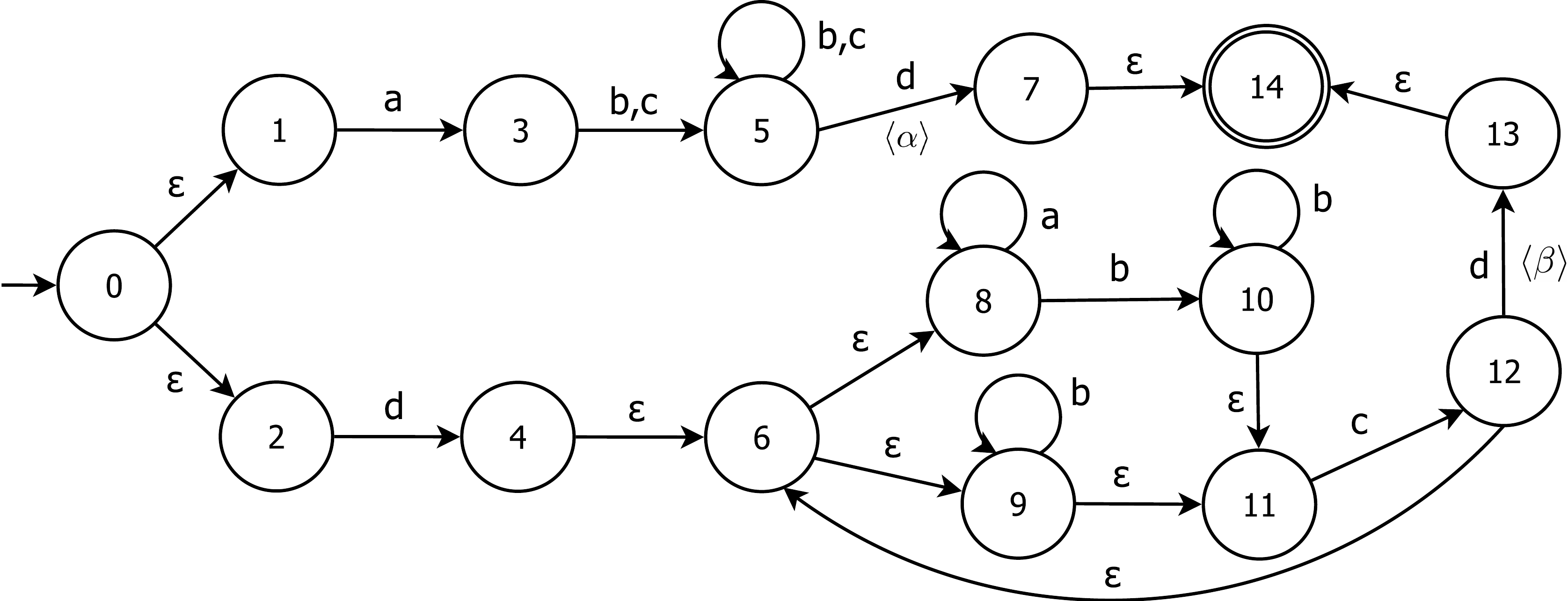}
 \end{center}
  \caption{An FST for 
  $e_3 = e_1 + e_2 = a(b+c)^+d\langle\alpha\rangle + d((a^*b^++b^*)c)^+d \langle\beta\rangle$.}
  \label{app:fig:fst_e3}
\end{figure}

\newpage

\bibliography{references}
\bibliographystyle{plain}

\end{document}

%% file: defns.tex
\newcommand{\trans}[2]{\stackrel[{\scriptscriptstyle\langle #2 \rangle}]{#1}{\longrightarrow}}
\newcommand{\transNO}[1]{\stackrel{#1}{\longrightarrow}}
\newcommand{\transNOstar}[1]{\stackrel{#1}{\longrightarrow}_{\varepsilon^*}}

\newcommand{\M}{\mathcal{M}}
\newcommand{\A}{\mathcal{A}}
\newcommand{\Parts}{\mathcal{P}}
\newcommand{\Lang}{\mathcal{L}}
\newcommand{\B}{\mathcal{B}}
\newcommand{\Bo}{\mathcal{B}_{\Gamma}}
\newcommand{\Out}{\mathcal{O}}
\newcommand{\minimize}{\texttt{minComp}}

\newcommand{\AlphbPre}{A}
\newcommand{\Alphb}{U}
\newcommand{\alp}{u}
\newcommand{\letter}{\tau}

\newcommand*{\inlineequation}[2][]{%
  \begingroup
    \refstepcounter{equation}%
    \ifx\\#1\\%
    \else
      \label{#1}%
    \fi
    \relpenalty=10000 %
    \binoppenalty=10000 %
    \ensuremath{%
      #2%
    }%
    ~\@eqnnum
  \endgroup
}

\newtheorem{proposition}{Proposition}[section]
\newtheorem{definition}{Definition}[section]

%% file: paper.bbl
\begin{thebibliography}{10}

\bibitem{Aho2006}
Alfred~V. Aho, Monica~S. Lam, Ravi Sethi, and Jeffrey~D. Ullman.
\newblock {\em Compilers: Principles, Techniques, and Tools (2nd Edition)}.
\newblock Addison-Wesley Longman Publishing Co., Inc., USA, 2006.

\bibitem{berstel}
Jean Berstel.
\newblock Transductions and context-free languages.
\newblock In {\em Teubner Studienb{\"u}cher : Informatik}, 1979.

\bibitem{berstel2010}
Jean Berstel, Luc Boasson, Olivier Carton, and Isabelle Fagnot.
\newblock Minimization of automata, 2010.

\bibitem{Bispo2007}
Joao Bispo, Ioannis Sourdis, João Cardoso, and Stamatis Vassiliadis.
\newblock Regular expression matching for reconfigurable packet inspection.
\newblock pages 119 -- 126, 01 2007.

\bibitem{church_1958}
Alonzo Church.
\newblock Edward f. moore. gedanken-experiments on sequential machines.
  automata studies, edited by c. e. shannon and j. mccarthy, annals of
  mathematics studies no. 34, litho-printed, princeton university press,
  princeton1956, pp. 129–153.
\newblock {\em Journal of Symbolic Logic}, 23(1):60–60, 1958.

\bibitem{ceska2019}
Milan Češka, Vojtěch Havlena, Lukáš Holík, Jan Kořenek, Ondřej Lengál,
  Denis Matoušek, Jiří Matoušek, Jakub Semrič, and Tomáš Vojnar.
\newblock Deep packet inspection in fpgas via approximate nondeterministic
  automata, 2019.

\bibitem{elgot}
C.~C. {Elgot} and J.~E. {Mezei}.
\newblock On relations defined by generalized finite automata.
\newblock {\em IBM Journal of Research and Development}, 9(1):47--68, Jan 1965.

\bibitem{Hopcroft1971}
J.~Hopcroft.
\newblock An n log n algorithm for minimizing states in a finite automaton.
\newblock 1971.

\bibitem{Hopcroft1979}
John~E. Hopcroft.
\newblock {\em Introduction to automata theory, languages, and computation}.
\newblock Addison-Wesley series in computer science. Addison-Wesley, Reading,
  Mass., 1979.

\bibitem{Johnson1968}
Walter~L. Johnson, James~H. Porter, Stephanie~I. Ackley, and Douglas~T. Ross.
\newblock Automatic generation of efficient lexical processors using finite
  state techniques.
\newblock {\em Commun. ACM}, 11(12):805–813, December 1968.

\bibitem{Kaplan1994}
Ronald~M. Kaplan and Martin Kay.
\newblock Regular models of phonological rule systems.
\newblock {\em Comput. Linguist.}, 20(3):331–378, September 1994.

\bibitem{Klimovich2012}
A.~S. Klimovich and V.~V. Solov'ev.
\newblock Minimization of mealy finite-state machines by internal states
  gluing.
\newblock {\em Journal of Computer and Systems Sciences International},
  51(2):244--255, Apr 2012.

\bibitem{Klimowicz2020}
Adam Klimowicz.
\newblock Combined state splitting and merging for implementation of fast
  finite state machines in fpga.
\newblock In Khalid Saeed and Ji{\v{r}}{\'i} Dvorsk{\'y}, editors, {\em
  Computer Information Systems and Industrial Management}, pages 65--76, Cham,
  2020. Springer International Publishing.

\bibitem{Louden1997}
Kenneth~C. Louden.
\newblock {\em Compiler Construction: Principles and Practice}.
\newblock PWS Publishing Co., USA, 1997.

\bibitem{mealy}
G.~H. {Mealy}.
\newblock A method for synthesizing sequential circuits.
\newblock {\em The Bell System Technical Journal}, 34(5):1045--1079, Sep. 1955.

\bibitem{mohri}
Mehryar Mohri.
\newblock Finite-state transducers in language and speech processing.
\newblock {\em Comput. Linguist.}, 23(2):269--311, June 1997.

\bibitem{Mohri2005}
Mehryar Mohri, Fernando Pereira, and Michael Riley.
\newblock Weighted automata in text and speech processing, 2005.

\bibitem{Rodrigues2019}
João Rodrigues, Duarte Folgado, David Belo, and Hugo Gamboa.
\newblock Ssts: A syntactic tool for pattern search on time series.
\newblock {\em Information Processing \& Management}, 56(1):61--76, 2019.

\bibitem{Solovev2011}
V.~V. Solov'ev.
\newblock Minimization of mealy finite state machines via internal state
  merging.
\newblock {\em Journal of Communications Technology and Electronics},
  56(2):207--213, Feb 2011.

\bibitem{Thompson1968}
Ken Thompson.
\newblock Programming techniques: Regular expression search algorithm.
\newblock {\em Commun. ACM}, 11(6):419–422, June 1968.

\end{thebibliography}
